\documentclass[11pt,a4paper]{article}

\usepackage[UKenglish]{babel}
\usepackage[T1]{fontenc}
\usepackage[utf8]{inputenc}
\usepackage{lmodern}
\usepackage[margin=1in]{geometry}
\usepackage{microtype}
\usepackage{graphicx}
\usepackage{amsmath,amssymb,amsfonts,mathtools,amsthm}
\usepackage{stmaryrd}
\usepackage{array,booktabs,multirow,tabularx}
\usepackage{xcolor}
\usepackage{enumitem}
\usepackage{todonotes}
\usepackage{tikz}
\usetikzlibrary{arrows.meta,positioning,calc,shapes.geometric,
  decorations.pathreplacing,decorations.pathmorphing,fit,backgrounds,
  matrix,chains,quotes}
\usepackage{framed}
\usepackage{etoolbox}
\usepackage{hyperref}
\usepackage[capitalise,noabbrev]{cleveref}

\hypersetup{
  unicode=true,
  colorlinks=true,
  linkcolor=blue,
  citecolor=blue,
  urlcolor=blue,
  pdfauthor={Manfred Droste and Guo-Qiang Zhang},
  pdftitle={Discrete Linear Ensemble Logic: Decidability, Expressiveness, and Axiomatization}
}

\setlist{nosep,leftmargin=*}
\theoremstyle{plain}
\newtheorem{theorem}{Theorem}
\newtheorem{lemma}[theorem]{Lemma}
\newtheorem{proposition}[theorem]{Proposition}
\newtheorem{corollary}[theorem]{Corollary}

\theoremstyle{definition}
\newtheorem{definition}[theorem]{Definition}
\newtheorem{example}[theorem]{Example}

\makeatletter
\let\@articlekeywords\@empty
\let\@articleclassifications\@empty
\newcommand{\keywords}[1]{\gdef\@articlekeywords{#1}}
\newcommand{\ccsdesc}[2][]{%
  \ifx\@articleclassifications\@empty
    \gdef\@articleclassifications{#2}%
  \else
    \g@addto@macro\@articleclassifications{; #2}%
  \fi}
\newcommand{\printarticlemetadata}{%
  \begin{center}
  \begin{minipage}{0.94\textwidth}
  \small
  \ifx\@articlekeywords\@empty\else
    \noindent\textbf{Keywords:} \@articlekeywords\par
  \fi
  \ifx\@articleclassifications\@empty\else
    \medskip
    \noindent\textbf{Subject classifications:} \@articleclassifications\par
  \fi
  \end{minipage}
  \end{center}
  \medskip}
\AfterEndEnvironment{abstract}{\printarticlemetadata}
\makeatother

\definecolor{routeblue}{RGB}{34,74,148}
\definecolor{routegreen}{RGB}{35,120,70}
\definecolor{routered}{RGB}{160,45,45}
\definecolor{routepurple}{RGB}{102,65,145}
\definecolor{routeorange}{RGB}{196,112,35}
\definecolor{softgrey}{RGB}{246,247,249}
\definecolor{ledgergold}{RGB}{132,96,16}
\definecolor{deepink}{RGB}{32,39,55}
\definecolor{alphabetteal}{RGB}{0,112,126}

\definecolor{revblue}{rgb}{0.0,0.25,0.7}
\definecolor{addgreen}{rgb}{0.0,0.45,0.15}
\definecolor{auditgold}{rgb}{1.0,0.97,0.85}

\definecolor{shadecolor}{rgb}{1.0,0.97,0.85}

\newcommand{\mdiamond}{\mathord{\scalebox{1.05}{$\Diamond$}}\mkern-2mu}
\newcommand{\boldBox}{\mathord{\scalebox{1.05}{$\Box$}}\mkern-0.5mu}
\newcommand{\EL}{\textsf{EL}}
\newcommand{\TPTL}{\textsf{TPTL}}
\newcommand{\MTL}{\textsf{MTL}}

\newcommand{\HS}{\textsf{HS}}
\newcommand{\LTL}{\textsf{LTL}}
\newcommand{\FOLTL}{\textsf{FO\hbox{-}LTL}}
\newcommand{\FO}{\mathrm{FO}}
\newcommand{\MSO}{\mathrm{MSO}}

\newcommand{\Nat}{\mathbb{N}}
\newcommand{\Natp}{\mathbb{N}_{>0}}
\newcommand{\HEL}{\mathcal{H}_{\EL}}
\newcommand{\Th}{\mathrm{Th}}

\newcommand{\MONFO}{\mathrm{MONFO}_{\Nat,+}}

\newcommand{\event}[1]{\mathsf{#1}}

\title{Discrete Linear Ensemble Logic:\\
       Decidability, Expressiveness, and Axiomatization}
\author{Manfred Droste$^{1}$ \qquad Guo-Qiang Zhang$^{2}$\\[0.75em]
\small $^{1}$Institute of Computer Science, Leipzig University, Leipzig, Germany\\
\small $^{2}$The University of Texas Health Science Center at Houston, Houston, Texas, USA\\
\small Corresponding to \texttt{guo-qiang.zhang@uth.tmc.edu}}
\date{}

\begin{document}
\maketitle

\begin{abstract}
We study the discrete point-based fragment of Ensemble Logic $\EL(\Nat)$ over the natural numbers, a logic combining displacement $\varphi_u$, bounded metric modalities $\boldBox_t$ and $\mdiamond_t$ with additive bounds, Boolean connectives, and first-order quantification over $\Nat$.  Motivated by the need for a unified symbolic layer for biomedical knowledge with temporal, spatial, genomic, and multimodal metric content, we develop the foundational discrete theory of the formalism.  We give syntax and semantics, and prove a forward embedding of $\EL(\Nat)$ over a finite proposition set $\mathcal{P}$ into first-order monadic Presburger arithmetic $\FO(\Nat,<,+;\mathcal{P})$.  This embedding yields the analytical upper bounds, while a reduction from nondeterministic two-counter machines with recurring control states proves that satisfiability is $\Sigma^1_1$-complete and validity is dually $\Pi^1_1$-complete.  Expressively, $\EL(\Nat)$ strictly extends the star-free $\omega$-languages and is incomparable with the $\omega$-regular languages: it defines the non-$\omega$-regular counting language $\{a^mb^mc^md^m\mid m\geq 1\}\cdot\Sigma^\omega$, whereas a delimited parity language remains outside the logic by classical Presburger-arithmetic lower bounds.  On the proof-theoretic side, we present a sound Hilbert system $\HEL$ and establish completeness relative to monadic Presburger validity as oracle, noting that completeness relative to plain Presburger arithmetic is impossible.  We also prove that the existential fragment $\exists\EL(\Nat)$ has $\mathrm{NP}$-complete satisfiability and $\mathrm{coNP}$-complete unsatisfiability. Finally, we show that finite active-domain model checking has $\mathrm{PTIME}$ data complexity and $\mathrm{PSPACE}$-complete combined complexity.  Together, these results give a precise decidability, expressiveness, proof-theoretic, and model-checking baseline for subsequent algorithmic applications of Ensemble Logic in biomedicine.
\end{abstract}

\section{Introduction}\label{sec:intro}

Modern biomedicine has a formalism gap.  Taxonomical ontologies such as SNOMED CT and
UMLS~\cite{bodenreider2004umls} capture terminological structure but not metric time; pathway languages such as
SBML~\cite{hucka2003sbml} and BioPAX~\cite{demir2010biopax}
capture biochemical network structure but not spatial containment;
standard temporal logics~\cite{ltl} 
capture traces but not anatomical geometry; mathematical
morphology captures shapes but not protocol-level quantifier scope
\cite{serra1982morphology}.
A clinical constraint such as an HbA1c normalization condition that must hold
continuously for $24$ weeks, or a lymph-node tumour deposit large enough to
contain a $1$\,mm spherical core, has a clear mathematical truth condition.  
Yet neither condition maps naturally to a single formal representation without
external scripting in a programming setting.

The central observation behind Ensemble Logic ($\EL$)~\cite{zhang2024temporal} is that
biomedical reasoning repeatedly uses three non-Boolean operations: exact displacement,
bounded existence, and bounded universality.  In the discrete setting these are written
$\varphi_u$, $\mdiamond_t\varphi$, and $\boldBox_t\varphi$ ({\bf Fig.~\ref{fig:orientation}}).  Together with Boolean
connectives and first-order quantifiers over an additive index domain, they yield the
syntax
\begin{equation}\label{eq:syntax-intro}
\varphi,\psi\;::=\;p\mid\varphi_u\mid\neg\varphi\mid\varphi\wedge\psi\mid\varphi\vee\psi
\mid\mdiamond_t\varphi\mid\boldBox_t\varphi\mid\exists x\,\varphi\mid\forall x\,\varphi.
\end{equation}

The interpretation domain of $\EL$ accommodates discrete time, dense continuous time, Euclidean space, genomic coordinates, 
or mixed product spaces depending on the underlying monoid interpretation of the additive terms $u$ and $t$~\cite{zhang2024temporal}, which in general come from different monoids.

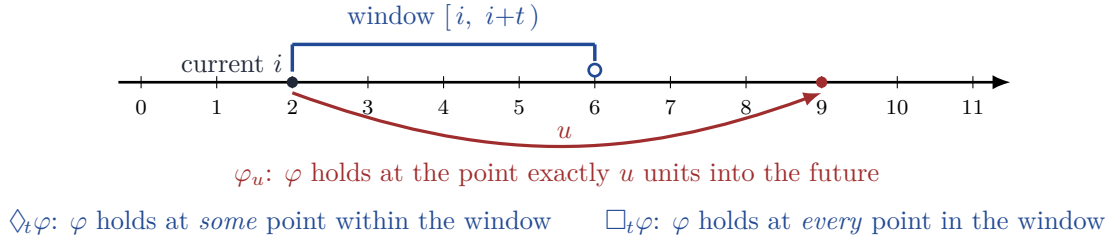
\begin{figure}[ht]
\centering
\begin{tikzpicture}[>={Latex[length=2.4mm,width=1.9mm]},font=\small,
  pt/.style={circle,fill=deepink,inner sep=1.5pt},
  rpt/.style={circle,fill=routered,inner sep=1.5pt}]
\draw[->,line width=1pt] (-0.3,0)--(11.5,0);
\foreach \x in {0,...,11}{\draw (\x,0.08)--(\x,-0.08);
  \node[below=0.5pt] at (\x,-0.08){\scriptsize \x};}
\node[pt] (i) at (2,0){};
\node[routered, below=13pt] at (5.6,0){$u$};
\node[deepink,above left=0pt and -3pt] at (i){current $i$\;};
\draw[line width=1.3pt,routeblue] (2,0.5)--(6,0.5);
\draw[line width=1.3pt,routeblue] (2,0.5)--(2,0.14);
\draw[line width=1.3pt,routeblue] (6,0.5)--(6,0.28);
\draw[routeblue,fill=white,line width=1pt] (6,0.16) circle (2.3pt);
\node[routeblue,above] at (4,0.55){window $[\,i,\ i{+}t\,)$};
\draw[->,line width=1.2pt,routered,bend right=20] (2,-0.14) to
  node[below=2pt]{$\varphi_u$: $\varphi$ holds at the point exactly $u$ units into the future} (9,-0.14);
\node[rpt] at (9,0){};
\node[routeblue, anchor=north,align=center] at (5.5,-1.55)
  {$\mdiamond_t\varphi$: $\varphi$ holds at \emph{some} point within the window\qquad
   $\boldBox_t\varphi$: $\varphi$ holds at \emph{every} point in the window};
\end{tikzpicture}
\caption[The metric vocabulary: windows and the jump $\varphi_u$]{An illustration of how the three operators work along the integer line.}
\label{fig:orientation}
\end{figure}

\paragraph*{Contributions}

In this paper, we focus the mathematical investigation on the foundational discrete case $\EL(\Nat)$, originally introduced in~\cite{zhang2024temporal}, where $\Nat$ is interpreted temporally (modeling disease trajectories) or spatially (such as base coordinates in genomics).
After brief motivational overview and comparison with existing temporal logic formalisms, we
 define the syntax and semantics of $\EL(\Nat)$ (\S\ref{sec:syntax}). 
 We show that satisfiability of $\EL(\Nat)$ formulas is undecidable, and in fact $\Sigma^1_1$-hard, by a reduction from the recurring problem for nondeterministic two-counter machines.
  Building on the foundation laid out in~\cite{zhang2024temporal},  this paper establishes results in the directions of decidability, expressiveness, and 
  axiomatization, with specific topics that lie at the  heart of any new logical formalism:
  
 \begin{enumerate}
 \item We translate $\EL(\Nat)$ into monadic Presburger arithmetic, obtaining the analytical upper bounds $\Sigma^1_1/\Pi^1_1$ (\S\ref{sec:forward}). 
  \item We further prove $\Sigma^1_1$-hardness by a two-counter-machine reduction tailored to $\EL$'s uniformly positive modal bounds (\S\ref{sec:undec-2cm}). 
   
 \item  We place $\EL(\Nat)$ in the descriptive hierarchy: it contains all $\FO(\Nat,<;\mathcal P)$ languages and defines non-$\omega$-regular languages (\S\ref{sec:hierarchy}). 
 \item  We give a sound Hilbert system and prove $\MONFO$-relative completeness (\S\ref{sec:hilbert}). 
    \item We prove NP-completeness of existential-fragment satisfiability, coNP-completeness of unsatisfiability, and give a complete proof system for satisfiability (\S\ref{sec:exists}).
  \item We prove finite active-domain model checking is in PTIME for fixed formulas and PSPACE-complete (\S\ref{sec:mc}). 
 
\end{enumerate}

\section{Background}\label{sec:back}

\subsection{Motivations from Biomedicine}\label{subsec:bio-motivation}

Modern biomedicine depends on knowledge expressed across free-text discovery claims,
longitudinal records, trial protocols, physiological signals, imaging, and omics.  These modalities are still analyzed with fragmented formalisms and ad hoc scripts.  The same
patient, lesion, molecular locus, or protocol event is therefore represented differently
across literature mining, EHR phenotyping, trial eligibility, image analysis, and
genomics, even when the underlying scientific claim has the same structure.

In prior work~\cite{time25,scaleup}, we demonstrated that the tri-modality of
$\varphi_u$, $\mdiamond_t\varphi$, and $\boldBox_t\varphi$ (together with standard Boolean and first-order quantifiers) suffices to articulate the entire temporal structure of clinical-trial protocols in a single closed formula, such as one for Alzheimer's-disease trial NCT01266525~\cite{nct01266525}. 
We illustrate the key modeling idea behind these works using a deliberately compact monitoring protocol example.

\begin{quote}
Whenever therapy starts on day $x$, a baseline measurement must occur during the
four-day window beginning at $x$; a dose (of medication) must be taken on every day of the
fourteen-day window beginning one day later; and an outcome review must occur
exactly twenty-eight days after the start.
\end{quote}

With propositions $\event{Start}$, $\event{Baseline}$, $\event{Dose}$, and
$\event{Review}$, the protocol is expressed as
\begin{equation}\label{eq:guiding-rule}
\forall x\Bigl(
  \event{Start}_x \to
  \bigl((\mdiamond_4\event{Baseline})_x
  \wedge (\boldBox_{14}\event{Dose})_{x+1}
  \wedge \event{Review}_{x+28}\bigr)
\Bigr).
\end{equation}
The half-open-window convention makes the three clauses precise:
\begin{center}
\small
\begin{tabularx}{0.94\textwidth}{lll}
\toprule
\textbf{Fragment} & \textbf{Plain reading} & \textbf{Truth condition} \\
\midrule
$\event{Review}_{x+28}$
& exact displacement
& $\event{Review}$ takes place on day $x+28$; \\
$(\mdiamond_4\event{Baseline})_x$
& bounded existence
& $\event{Baseline}$ holds on day $x$, $x+1,x+2,$ or day $x+3$; \\
$(\boldBox_{14}\event{Dose})_{x+1}$
& bounded universality
& $\event{Dose}$ is taken on every day in $[x+1,x+15)$. \\
\bottomrule
\end{tabularx}
\end{center}
The variable $x$ is not a patient variable or an event object.  It is a numerical
offset from the origin at which the closed formula is evaluated.  Thus
$\forall x$ checks every possible start position in the record.

Quantification over metric parameters and logical inference become useful when the same unknown
length must be reused.  For example, the schematic rule
\begin{equation}\label{eq:reused-bound}
\forall x\Bigl(\event{Start}_x\to
 \exists n\bigl[
 \boldBox_{n+1}\event{Stable} \wedge
 (\boldBox_{n+1}\event{Treatment})_{n+1}\wedge
 (\mdiamond_{n+1}\event{Outcome})_{2(n+1)}
 \bigr]_x\Bigr)
\end{equation}
asks for a positive, record-dependent phase length $n+1$ and reuses exactly that
length for a stabilization phase, a treatment phase, and an outcome window.  This
reuse of a semantically chosen additive parameter is the most distinctive syntactic  feature of $\EL$.

A complete clinical-trial protocol is naturally a conjunction of many clauses
of the form illustrated by~\eqref{eq:guiding-rule}, together with domain-specific
eligibility predicates and data-quality assumptions.  Prior work shows how such
clauses can be composed for trial protocols \cite{time25}; the present paper is
concerned with the underlying discrete logic, not with validating any particular
medical protocol.

The same approach extends to other medical domains.
In sleep medicine~\cite{sleep},
the formula $(\varphi_{t})$ places a finding $\varphi$ at a specific offset from a reference time point;  
    for example, snore occurs 6 seconds later can be formalized as $\event{Snore}_{6}$.
 Similarly, $(\Box_t \varphi)$ states that the condition $\varphi$ holds continuously for duration $t$; 
    for example, airflow reduction greater than 30\% persisting for 10 seconds 
    can be formalized as $\Box_{10}\event{FlowDrop30}$. 
    $(\Diamond_t \varphi)$ states that  property $\varphi$ holds somewhere within a bounded window of length $t$, so that a spindle occurring within a 3-second window can be formalized as $\Diamond_3\event{Spindle}$.

The scientific content with respect to how long, how far, how soon, how continuously,
is exactly the metric structure lost by mainstream approaches.
The EL thesis~\cite{zhang2024temporal}
is that displacement, constrained existence, and bounded universality are
parts of a stable ``lingua franca'' for biomedical knowledge.  The novelty is not that any one
operator is new.  Linear temporal logic, metric temporal logic, and signal temporal logic already contain important special cases (see Section~\ref{back:logics}). 
The novelty is the cross-modal
stability of the operator basis: the same tri-modality can be interpreted on sequences,
claim windows, EHR days, dense waveforms, Euclidean image regions, stranded genomic
coordinates, and product spaces for spatial-temporal interactions.

\subsection{Comparison with neighboing formalisms}\label{back:logics}

The formalisms in Table~\ref{tab:paradigms} are selected according to three
principled axes: the semantic carrier (points, intervals, or signals), the way a
metric parameter is supplied (fixed constant, frozen timestamp, or quantified
offset), and the domain over which first-order variables range.  They are not
intended as an exhaustive survey, but as an illustration of how EL thesis~\cite{zhang2024temporal}
is genuinely distinct from existing formalisms.

Propositional linear temporal logic ($\LTL$) is the classical point-based logic of
linear executions: formulas are evaluated at positions of a trace, and temporal
operators such as next $\textsf{X}$ and until $\mathcal{U}$ describe how truth
evolves along that trace~\cite{ltl,kamp}. $\LTL$ is
qualitative: it can express that an event eventually occurs, but not that it
occurs within a specified numerical delay.

Timed and metric temporal logics add quantitative time in different ways. Metric
temporal logic ($\MTL$~\cite{koymans,ouaknine-worrell}) decorates temporal operators with fixed intervals, as in
$\varphi\,\mathcal{U}_{[a,b]}\,\psi$. Timed
propositional temporal logic ($\TPTL$~\cite{alur-henzinger}) instead introduces freeze variables: a
formula may store the current timestamp and later compare it with other
timestamps by arithmetic constraints~\cite{alur-henzinger}. 
 First-order
$\LTL$ combines temporal operators with quantification over an external data
domain, while Halpern-Shoham logic ($\HS$) evaluates formulas over intervals
rather than points, using Allen-style interval relations~\cite{allen,halpern1991propositional}.

The logic $\EL$ is closest in spirit to metric temporal logics, but its metric
parameters are not merely constants written into the syntax. Instead, $\EL$
quantifies over additive displacements of the timeline. Thus a formula can choose
an offset $n$ and then reason about positions such as $v+n$ or $v+n+1$. In
discrete time, a bounded metric-until formula has the $\EL$-style rendering
$\varphi\,\mathcal{U}_{[a,b]}\,\psi
  \;\equiv\;
  \bigvee_{r=a}^{b}
  \bigl(\psi_r \wedge \boldBox_{r}\varphi\bigr),$
up to the usual endpoint convention for until. Conversely, $\EL$ can quantify
over the bound itself, as in formulas of the form
$\exists n\,(\boldBox_{n+1}\varphi \wedge \cdots)$, which is outside ordinary
fixed-interval $\MTL$. This ability to quantify over an unbounded domain of
indices is also what conceptually brings $\EL$ closer to highly undecidable
interval logics such as Halpern-Shoham logic over discrete linear orders
~\cite{halpern1991propositional,bresolin2011}.
The table below summarizes key aspects of distinct logical formalisms.

\begin{table}[t]
\caption{Comparison along semantic carrier, metric mechanism, and quantified domain.}
\label{tab:paradigms}
\centering
\small
\begin{tabularx}{\textwidth}{@{}lXXX@{}}
\toprule
\textbf{Logic} & \textbf{Semantic carrier} & \textbf{Metric mechanism} &
\textbf{Quantified domain} \\
\midrule
$\EL(\Nat)$ & point of an $\omega$-word & exact shifts and bounded windows whose
additive parameters may contain variables & numerical offsets in $\Nat$ \\
$\LTL$ & point of an infinite trace & next and qualitative until; no numerical
parameter in the standard syntax & none in propositional $\LTL$ \\
$\MTL$ & timed position or time point & syntactically fixed metric intervals
$I$ on temporal operators & normally no quantification over interval bounds \\
$\TPTL$ & timed-word position with timestamp & freeze binder plus clock
constraints & frozen timestamps \\ 
$\FOLTL$ & trace position carrying a first-order structure & qualitative temporal
operators plus first-order atoms & external objects, not temporal offsets \\
\bottomrule
\end{tabularx}
\end{table}

\paragraph*{Attributes of $\EL$: Natural and Minimalist}

An important distinction between $\EL$ and the other formalisms mentioned above lies in its natural and minimalist character.
A property of an infinite word is a yes/no question one can ask about the sequence $\alpha(0),\alpha(1),\alpha(2), \dots$.  
Classical temporal logic $\LTL$ answers such questions by stepping through the word one position at a time, using operators such
as ``next'' and ``until;''  it can look arbitrarily far ahead, but it can only \emph{count} by chaining single steps.  

Ensemble Logic keeps the temporal intuition of a point moving along the word, but it hands
that point a ruler ({\bf Fig.~\ref{fig:orientation}}).  Instead of ``sometime later'' it can say ``somewhere within the next $t$ positions,'' and instead of ``the next
position'' it can say ``the position exactly $u$ steps further from here'' (also called ``displacement'')
where $t$ and $u$ are \emph{additive terms} built from numerical variables (integer variables in this paper).  
The single design decision that gives the logic its character is that these lengths are first-class objects: one variable
may fix a distance that is then reused, combined, and added elsewhere in the same formula.

This is also the reason that complexity results based on $\LTL$ do not directly apply or need to be revisited in the $\EL$ context: an algorithm linear in the size of $\LTL$ formulas can 
only be directly translated to one that it exponential to the size of $\EL$ formulas.

Natural and minimalist are also reflected in $\EL$'s treatment of biomedical data records directly as semantic models (with existing annotated features). 
In contrast, other formalisms usually require a more elaborate semantic structure to work with. 
However, a full exposition of how biomedical data records directly correspond to semantic models for $\EL$  is beyond the scope of the current paper.

\section{Syntax and Semantics}\label{sec:syntax}

We first recall basic definitions from~\cite{zhang2024temporal}.
Throughout, $\Nat=\{0,1,2,\ldots\}$ and $\Natp=\{1,2,\ldots\}$. We adopt a uniform
convention: all first-order variables range over $\Nat$, 
while modal lengths must evaluate to numbers $\geq 1$.

\begin{definition}[Terms]\label{def:terms}
We fix a countably infinite set of variables $V=\{x,y,z,\ldots\}$. 
\emph{Additive terms} are defined as
$u,v,t,s\;::=\;n\;\mid\;x\;\mid\;u+v,$
where $n\in\Nat$ is a natural-number numeral and $x\in V$. 
Numerals are encoded in binary. The syntactic size of a numeral $n$ is bounded by $|n|=\Theta(\log(n+1))$, and the size $|u|$ of a term is given by  the size of its syntax tree under this convention. Throughout the paper, $|\varphi|$ 
denotes the binary-encoded length of $\varphi$.

An assignment $\eta:V\to\Nat$ naturally extends to terms by homomorphic evaluation: $\eta(n)=n$ and $\eta(u+v)=\eta(u)+\eta(v)$. Whenever a term $t$ appears as the metric subscript of a modality, its evaluation must be strictly positive in the current assignment ($\eta(t) \geq 1$). Syntactically, this can be ensured by writing modality bounds as $t+1$.
\end{definition}

\begin{definition}[Formulas]\label{def:formulas}
For a finite set $\mathcal{P}$ of atomic propositions, the formulas of $\sf{EL}(\Nat)$ are generated by the grammar in~\eqref{eq:syntax-intro}, subject to the condition that every closed instance of a modal length evaluates to $\geq 1$. Standard abbreviations for $\top,\bot,\to,\leftrightarrow$ are derived as usual.
\end{definition}

An \emph{interpretation} is defined as $\alpha: \Nat\to 2^{\mathcal{P}}$, equivalently an \emph{$\omega$-word} over the alphabet $2^{\mathcal{P}}$. We denote the support of a proposition $p$ by $P^\alpha_p=\{i\in\Nat \mid p\in\alpha(i)\}$.

\begin{definition}[Satisfaction Relation]\label{def:sat}
For a given interpretation $\alpha$, variable assignment $\eta$, and temporal/spatial origin point $i\in\Nat$, satisfaction is defined inductively by standard Boolean clauses paired with the metric primitives:
\begin{align*}
(\alpha,\eta,i)&\models p
   &\Leftrightarrow&\quad p\in\alpha(i),\\
(\alpha,\eta,i)&\models\varphi_u
   &\Leftrightarrow&\quad (\alpha,\eta,i+\eta(u))\models\varphi,\\
(\alpha,\eta,i)&\models\mdiamond_t\varphi
   &\Leftrightarrow&\quad \exists j\in[i,i+\eta(t))\;(\alpha,\eta,j)\models\varphi,\\
(\alpha,\eta,i)&\models\boldBox_t\varphi
   &\Leftrightarrow&\quad \forall j\in[i,i+\eta(t))\;(\alpha,\eta,j)\models\varphi,\\
(\alpha,\eta,i)&\models Qx\,\varphi
   &\Leftrightarrow&\quad Q n\in\Nat\;(\alpha,\eta[x\mapsto n],i)\models\varphi
   \quad (Q\in\{\exists,\forall\}).
\end{align*}

A formula is closed if it has no free first-order variables. Variables occurring
inside additive subscripts are free unless they occur in the syntactic scope of an enclosing first-order quantifier; displacement and metric modalities do not bind variables.
For a closed formula $\varphi$, $\alpha\models\varphi$ is a syntactic shorthand for $(\alpha,\emptyset,0)\models\varphi$, and the corresponding language is defined as $L(\varphi)=\{\alpha\in(2^{\mathcal{P}})^\omega \mid \alpha\models\varphi\}$.
\end{definition}

\begin{proposition}[\cite{zhang2024temporal}]\label{prop:basic-vals}
For all formulas $\varphi,\psi$ and additive terms $u,v,s,t$ evaluating positively at modal lengths, the following semantic equivalences hold:
\begin{enumerate}[label=\textbf{(\roman*)}, leftmargin=2em]
    \item $(\neg\varphi)_u\leftrightarrow\neg(\varphi_u)$, $(\varphi\wedge\psi)_u\leftrightarrow\varphi_u\wedge\psi_u$, and dually for $\vee$;
    \item $(\varphi_u)_v\leftrightarrow\varphi_{u+v}$;
    \item $\boldBox_t$ and $\mdiamond_t$ distribute over $\wedge$ and $\vee$ respectively;
    \item $\neg\boldBox_t\varphi\leftrightarrow\mdiamond_t\neg\varphi$ and dually for $\neg\mdiamond_t$;
    \item Monotonicity: if $\eta(s)\leq\eta(t)$, then $\boldBox_t\varphi\to\boldBox_s\varphi$ and $\mdiamond_s\varphi\to\mdiamond_t\varphi$.
\end{enumerate}
\end{proposition}

As an immediate structural corollary, every $\EL(\Nat)$ formula is equivalent to a strict negation-normal form (NNF) where $\neg$ occurs exclusively as a terminal operator on propositional atoms.

\section{Forward Embedding into Monadic Arithmetic}\label{sec:forward}

Let $\FO(\Nat,<,+;\mathcal{P})$ represent first-order logic over the standard structure $(\Nat,<,+)$ expanded by unary predicates $P_p$ for each $p\in\mathcal{P}$. Atomic formulas consist of $P_p(x)$, standard identity $x=y$, inequality $x<y$, and equalities for additive Presburger terms.

\begin{theorem}[Forward embedding]\label{thm:forward}
There exists a linear-time translation $T$ assigning to every $\sf{EL}(\Nat)$ formula
$\varphi$ and every arithmetic term $z$ denoting the current evaluation point an
arithmetical first-order formula $T_z(\varphi)$ such that, for every interpretation
$\alpha$, assignment $\eta$, and point $m\in\Nat$,
\[
(\alpha,\eta,m)\models\varphi
\;\Longleftrightarrow\;
\bigl(\Nat,<,+,(P_p^\alpha)_{p\in\mathcal{P}}\bigr),\eta[z\mapsto m]
\models T_z(\varphi),
\]
where bound variables are alpha-renamed to avoid capture of the distinguished variable
$z$ and of the fresh variables introduced by the modal clauses.
\end{theorem}

\begin{proof}
Define $T_z$ by structural recursion:
\begin{align*}
T_z(p)&:=P_p(z),\\
T_z(\neg\varphi)&:=\neg T_z(\varphi),\\
T_z(\varphi \,\mbox{op}\, \psi)&:=T_z(\varphi)\,\mbox{op}\,T_z(\psi),
   \quad (\mbox{op}\in\{\wedge,\vee\}),\\
T_z(\varphi_u)&:=T_{z+u}(\varphi),\\
T_z(Qx\,\varphi)&:=Qx\,T_z(\varphi),\quad(Q\in\{\exists,\forall\}),\\
T_z(\mdiamond_t\varphi)&:=\exists j\,(z\leq j\wedge j<z+t\wedge T_j(\varphi)),\\
T_z(\boldBox_t\varphi)&:=\forall j\,\bigl((z\leq j\wedge j<z+t)\to T_j(\varphi)\bigr),
\end{align*}
with $j$ fresh and with the harmless alpha-renaming convention just stated.
Correctness is by structural induction on $\varphi$. The atomic, Boolean, and
first-order quantifier cases are immediate from the semantics and the induction
hypothesis. The displacement clause uses
$(\alpha,\eta,m)\models\varphi_u$ iff
$(\alpha,\eta,m+\eta(u))\models\varphi$, which is exactly the induction hypothesis for
$T_{z+u}(\varphi)$. The modal clauses match the half-open interval
$[m,m+\eta(t))$ by quantifying over the fresh variable $j$ with
$z\le j<z+t$. Each constructor contributes a constant-size translation node when the
output is represented as a shared directed acyclic graph, so the translation is linear
in $|\varphi|$ under the binary encoding convention of Definition~\ref{def:terms}.
\end{proof}

\begin{corollary}[Analytical upper bound]\label{cor:upper}
Satisfiability of arbitrary closed $\sf{EL}(\Nat)$ formulas resides in $\Sigma^1_1$, and algorithmic validity resides in $\Pi^1_1$.
\end{corollary}

\begin{proof}
By Theorem~\ref{thm:forward}, the satisfiability of an $\EL$ formula $\varphi$ is equivalent to the satisfiability of the formula $T_0(\varphi)$. Hence deciding satisfiability means existentially quantifying the continuous array of possible monadic subset interpretations $P_p^\alpha$. This can be expressed by a second-order existential prefix $\exists (P_p)_{p \in \mathcal{P}}$ appended to an arithmetical first-order formula, bounding the procedure within $\Sigma^1_1$. Validity is obtained by the dual quantification $\forall (P_p)_{p \in \mathcal{P}}$, bounded in $\Pi^1_1$.
\end{proof}

\section{\texorpdfstring{$\Sigma^1_1$}{Sigma11}-Hardness via Two-Counter Machines}\label{sec:undec-2cm}

By Corollary \ref{cor:upper}, the $\EL(\Nat)$-satisfiability problem is in $\Sigma^1_1$. Here we use a reduction to the recurring problem for nondeterministic two-counter machines to show its $\Sigma^1_1$-hardness.
A nondeterministic two-counter machine (2CM)~\cite{minsky1967} comprises a finite control $(Q,q_0,q^*,I)$ alongside absolute counters $C_1,C_2\in\Nat$ driving functional increment and test-and-decrement branches. The recurring problem asks whether there is an infinite computational sequence starting from $(q_0,0,0)$ and containing the loop state $q^*$ infinitely often.

\begin{lemma}[Harel-Pnueli-Stavi~\cite{harel1983}]\label{lem:hps}
The recurring problem for nondeterministic two-counter machines is $\Sigma^1_1$-complete.
\end{lemma}

Now we can show:

\begin{theorem}[$\Sigma^1_1$-hardness]\label{thm:two-counter}
The satisfiability problem for closed $\sf{EL}(\Nat)$ formulas is $\Sigma^1_1$-hard.
\end{theorem}

\begin{proof}

Fix a nondeterministic two-counter machine $M=(Q,q_0,q^*,I)$.  We use propositions
$Q\cup\{c_1,c_2,e\}$ and first define exact letter formulas
\[
H_q:=q\wedge\bigwedge_{r\in Q\setminus\{q\}}\neg r\wedge\neg e\wedge\neg c_1\wedge\neg c_2,
\]
\[
E:=e\wedge\neg c_1\wedge\neg c_2\wedge\bigwedge_{r\in Q}\neg r,
\qquad
C_i:=c_i\wedge\neg e\wedge\neg c_{3-i}\wedge\bigwedge_{r\in Q}\neg r\quad(i=1,2).
\]
A configuration with state $q$ and counter values $(x,y)$ begins at position $p$ and is
encoded by
\[
\begin{aligned}
B_q(p,x,y):={}&(H_q)_p\wedge E_{p+1}\wedge(\boldBox_{x+1}C_1)_{p+2}\wedge E_{p+x+3}\\
&\wedge(\boldBox_{y+1}C_2)_{p+x+4}\wedge E_{p+x+y+5}.
\end{aligned}
\]
Thus a counter value $x$ is represented by a $C_1$-block of length $x+1$, and similarly
for $y$; all modal lengths are positive.  The next configuration begins immediately at
\[
N(p,x,y):=p+x+y+6.
\]
Let $\Phi_{\mathrm{init}}:=B_{q_0}(0,0,0)$.
For the transition formula, put the machine instructions into the usual finite normal
form: increments are always enabled; a decrement/test instruction for counter $C_i$ has
one branch for value $0$ and one branch for a positive value.  For each state $q$ and
each zero/positive case $\chi\in\{0,+\}^2$, write
$\bar x_0=0$, $\bar x_+=x'+1$, $\bar y_0=0$, and $\bar y_+=y'+1$.  The clause for
$q,\chi$ is
\[
\forall p\,\forall x'\,\forall y'\,
\Bigl(B_q(p,\bar x_\chi,\bar y_\chi)
\to\bigvee_{\iota\in I(q,\chi)}\mathrm{Succ}_\iota(p,\bar x_\chi,\bar y_\chi)\Bigr),
\]
where $I(q,\chi)$ is the finite set of instructions enabled in that case; an empty
disjunction is $\bot$.  The successor formula is obtained by writing the corresponding
next block at $N(p,\bar x_\chi,\bar y_\chi)$: for example,
\[
\mathrm{Succ}_{q:C_1:=C_1+1;q'}(p,x,y)=B_{q'}(N(p,x,y),x+1,y),
\]
and the positive decrement branch for $C_1$ is written with the current value $x'+1$
in the antecedent and successor value $x'$ in the consequent,
\[
B_q(p,x'+1,y)\to B_{q'}(N(p,x'+1,y),x',y).
\]
The zero branch is the analogous clause with current value $0$.  No arithmetic equality
atom is needed; the same variables are reused in the block positions and lengths.
Let $\Phi_{\mathrm{trans}}$ be the conjunction of these finitely many clauses.  Finally,
recurrence of $q^*$ is expressed by
\[
\Phi_{\mathrm{rec}}:=\forall r\,\exists s\,\exists x\,\exists y\,B_{q^*}(r+s,x,y).
\]
Set
$\varphi_M:=\Phi_{\mathrm{init}}\wedge\Phi_{\mathrm{trans}}\wedge\Phi_{\mathrm{rec}}$.

If $M$ has a recurring run, concatenate the corresponding exact blocks.  The initial,
transition, and recurrence clauses are then immediate.  Conversely, suppose
$\alpha\models\varphi_M$.  From $B_{q_0}(0,0,0)$ and $\Phi_{\mathrm{trans}}$, induction on
successor blocks yields an infinite sequence of adjacent blocks whose states and counter
values follow an enabled run of $M$.  Because the formulas $H_q,E,C_1,C_2$ are mutually
exclusive and every successor starts exactly at the end of the current block, no state
marker can occur inside a block of this chain; the adjacent blocks therefore tile
$\Nat$.  Hence every block satisfying $B_q(p,x,y)$ is a boundary of this run.  The
formula $\Phi_{\mathrm{rec}}$ says that such $q^*$-boundaries occur beyond every cutoff,
so the extracted run visits $q^*$ infinitely often.  Thus $\varphi_M$ is satisfiable iff
$M$ has a recurring run.  The construction uses one constant-size family of clauses per
state, instruction, and zero/positive case, and all numerals are fixed small constants;
hence $|\varphi_M|=\mathrm{poly}(|M|)$.
\end{proof}

\begin{corollary}[Analytical classification]\label{cor:complete}
The satisfiability problem for closed $\sf{EL}(\Nat)$ formulas is $\Sigma^1_1$-complete; the validity problem is $\Pi^1_1$-complete.
\end{corollary}

\begin{proof}
Upper limits are given by Corollary~\ref{cor:upper}. The hardness of the satisfiability problem is given by Theorem~\ref{thm:two-counter}. The hardness of the validity problems follows by considering $\neg\varphi_M$.
\end{proof}

\textbf{Remark.} A reduction from the Post Correspondence Problem to the satisfiability problem for closed $\EL(\Nat)$ formulas was given in~\cite{zhang2024temporal} earlier.
\section{Position in the Descriptive Hierarchy}\label{sec:hierarchy}

In this section we wish to compare the class of  $\EL$-definable languages with standard classes of  $\omega$-languages. As is well-known, the $\MSO(\Nat,<;\mathcal{P})$-definable languages coincide with the  $\omega$-regular languages, and
the $\FO(\Nat,<;\mathcal{P})$-definable languages coincide with the 
$\omega$-$\LTL$-definable and the $\omega$-star-free rational languages.

\begin{example}[A non-$\omega$-regular language  in $\EL$]\label{ex:abcd}
We consider the standard alphabet of letters  $\Sigma = \{a,b,c,d\}$. Since interpretations are subsets of
$\mathcal P$ rather than singleton alphabet symbols, using $a,b,c,d$ alone would also
accept positions where several propositions hold simultaneously. Therefore, we let
$\ell_a:=a\wedge\neg b\wedge\neg c\wedge\neg d$ and define
$\ell_b,\ell_c,\ell_d$ analogously.  The closed formula
\[
\exists n\Bigl[\boldBox_{n+1}\ell_a\wedge(\boldBox_{n+1}\ell_b)_{n+1}
   \wedge(\boldBox_{n+1}\ell_c)_{(n+1)+(n+1)}
   \wedge(\boldBox_{n+1}\ell_d)_{(n+1)+(n+1)+(n+1)}\Bigr]
\]
defines, under the standard singleton-letter encoding over
$\{a,b,c,d\}$, the language
$L_{\mathrm{abcd}}=\{a^mb^mc^md^m\mid m\geq 1\}\cdot\Sigma^\omega$ (choose target
index $m=n+1$).  By the pumping lemma, $L_{\mathrm{abcd}}
\notin\omega\text{-Reg}$.
\end{example}

\begin{example}[An $\omega$-regular language outside $\EL$]\label{ex:parity-not-EL}
Let $\Sigma: =\{a,b,c,d\}$, and $L_{\mathrm{par}}$ over $\Sigma$  be defined as
$L_{\mathrm{par}}:=\{wcd^\omega \mid w\in\{a,b\}^*,\ |w|_a\equiv 0\pmod 2\}.$
This language is $\omega$-regular: a two-state automaton records the parity of the
number of $a$'s before the unique delimiter $c$ and then checks the $d^\omega$ tail.
If $L_{\mathrm{par}}$ were definable in $\EL(\Nat)$, Theorem~\ref{thm:forward} would
make it definable in $\FO(\Nat,<,+;\mathcal P)$.  Restricting such a definition to
ultimately constant words of the form $wc d^\omega$ would yield an $\FO[<,+]$
definition of finite-word parity.  This contradicts the classical lower bound that
parity is not definable in first-order logic with addition on finite ordered words,
equivalently parity is not in uniform $\mathrm{AC}^0$~\cite{fss1984,ajtai1983sigma,hastad1986}.
\end{example}

Now we obtain:

\begin{theorem}[Hierarchy classification matrix]\label{thm:hierarchy} We have:\\

\textbf{(i)} $L(\FO(\Nat,<;\mathcal{P})) \subsetneq L(\sf{EL}(\Nat))\subseteq L(\FO(\Nat,<,+;\mathcal{P}))$.

\textbf{(ii)} $L(\sf{EL}(\Nat)) \setminus \omega\text{-Reg} \neq \emptyset$ and 
$\omega\text{-Reg} \setminus L(\sf{EL}(\Nat)) \neq \emptyset$.
\end{theorem}

\begin{proof}
\textbf{(i)} By Kamp's theorem~\cite{kamp}, $\FO(\Nat,<;\mathcal{P})$-definability of $\omega$-languages coincides with $\LTL$-definability. 
It is straightforward to translate $\LTL$-formulas into equivalent $\EL$-formulas (i.e., preserving the languages):
For the \emph{Next} operator, replace $\textsf{X} \varphi$  by  $\varphi_1$, and for \emph{strict Until},
replace  $\varphi U_{>0} \psi$  by  
$\psi_1 \vee \exists x (\psi_{x+2} \land (\boldBox_{x+1} \varphi)_1)$.
By Example~\ref{ex:abcd}, the inclusion is proper. 
For the second inclusion, we apply Theorem~\ref{thm:forward}.

\textbf{(ii)} Immediate by Examples~\ref{ex:abcd} and \ref{ex:parity-not-EL}.
\end{proof}

\section{A Hilbert System and Oracle-Relative Completeness}\label{sec:hilbert}

In this section, we wish to obtain a Hilbert system for $\EL$-logic and to consider completeness questions.
Absolute completeness is impossible by Corollary~\ref{cor:complete}; completeness relative
only to Presburger arithmetic is also impossible, since Presburger arithmetic is decidable
(Cooper~\cite{cooper1972}, Presburger~\cite{presburger1929}) and any r.e. system using only
$\Th(\Nat,<,+)$ as oracle would have an r.e. proof relation, contradicting the
$\Pi^1_1$-completeness of $\EL$-validity. Therefore we use the monadic Presburger
validity theory as oracle.

\paragraph*{The system $\HEL$.}
The axioms below are given as schemas; rules are modus ponens, universal generalisation,
and necessitation for $\varphi_u,\boldBox_t,\mdiamond_t$ from valid formulas. 

Group~I (propositional and first-order):  
\begin{description}
\item[\textbf{(A1)}] propositional tautologies;
\item[\textbf{(A2)}]
$\forall x\,\varphi(x)\to\varphi[u/x]$ when $u$ free for $x$; 
\item[\textbf{(A3)}]
$\forall x(\varphi\to\psi)\to(\varphi\to\forall x\psi)$ when $x\notin\mathrm{free}(\varphi)$;
\item[\textbf{(A4)}] all closed additive identities true in $(\Nat,<,+)$.
\end{description}

Group~II (displacement): 
\begin{description}
\item[\textbf{(B1)}] displacement distributes over $\wedge,\vee,\neg$;
\item[\textbf{(B2)}] $(\varphi_u)_v\leftrightarrow\varphi_{u+v}$;
\item[\textbf{(B3)}] if $u=v$ is arithmetically true, $\varphi_u\leftrightarrow\varphi_v$;
\item[\textbf{(B4)}] $(Qx\,\varphi)_u\leftrightarrow Qx\,(\varphi_u)$ for $Q\in\{\exists,\forall\}$,
$x\notin\mathrm{free}(u)$.
\end{description}

Group~III (metric modalities):
\begin{description}
\item[\textbf{(C1)}] $\neg\boldBox_t\varphi\leftrightarrow\mdiamond_t\neg\varphi$ (dually);
\item[\textbf{(C2)}] $\boldBox_t(\varphi\to\psi)\to(\boldBox_t\varphi\to\boldBox_t\psi)$;
\item[\textbf{(C3)}] distribution of $\boldBox_t/\mdiamond_t$ over $\wedge/\vee$;
\item[\textbf{(C4)}] monotonicity in arithmetic length;
\item[\textbf{(C5)}] $\boldBox_1\varphi\leftrightarrow\varphi$, $\mdiamond_1\varphi\leftrightarrow\varphi$;
\item[\textbf{(C6)}] $(\boldBox_t\varphi)_u\leftrightarrow\boldBox_t(\varphi_u)$ (dually);
\item[\textbf{(C7)}] for $s,t\geq 1$, nested modalities compose:
$\mdiamond_s\mdiamond_{t+1}\varphi\leftrightarrow\mdiamond_{s+t}\varphi$ (dually);
\item[\textbf{(C8)}] \emph{additive decomposition} for $\eta(u),\eta(v)\geq 1$:
$\mdiamond_{u+v}\varphi\leftrightarrow\mdiamond_u\varphi\vee(\mdiamond_v\varphi)_u$
(dually). 
\end{description}

Group~IV (induction over displacement variables): 
\begin{description}
\item[\textbf{(D1)}] 
$\exists x\,\varphi(x)\leftrightarrow\varphi(0)\vee\exists y\,\varphi(y+1)$;
 (dually for $\forall$); \item[\textbf{(D2)}] induction
$\varphi(0)\wedge\forall x(\varphi(x)\to\varphi(x+1))\to\forall x\,\varphi(x)$.
\end{description}

\begin{theorem}[Soundness]\label{thm:soundness}
If $\HEL\vdash\varphi$ then $\models\varphi$.
\end{theorem}

\begin{proof}
We verify the non-trivial axioms.
\textbf{C7:} $(\alpha,\eta,i)\models\mdiamond_s(\mdiamond_t\varphi)$ iff
$\exists r<s,q<t.\,(\alpha,\eta,i+r+q)\models\varphi$; the attainable sums
$r+q$ form $[0,s+t-1)$, giving $\mdiamond_{s+t-1}\varphi$.

\smallskip\noindent\textbf{C8:} for $a=\eta(u),b=\eta(v)\geq 1$, use the partition
\[
[i,i+a+b)=[i,i+a)\cup[i+a,(i+a)+b).
\]
\smallskip\noindent\textbf{D1:} from $\Nat=\{0\}\cup\{y+1 \mid y \geq 0\}$; \textbf{D2} is mathematical induction.
Other axioms follow from Proposition~\ref{prop:basic-vals} and direct inspection.
\end{proof}

\paragraph*{The oracle}
\[
\begin{aligned}
\MONFO := \{\,\Theta\in\FO(\Nat,<,+;\mathcal{P}) \mid {}&
(\Nat,<,+,(P_p))\models\Theta\\
&\text{for every interpretation of the }P_p\,\}.
\end{aligned}
\]
This is the standard-validity theory of monadic Presburger arithmetic.

\begin{theorem}[Oracle-relative completeness]\label{thm:oracle}
Extend $\HEL$ with the rule: from $T_0(\varphi)\in\MONFO$ infer $\varphi$ (for closed
$\varphi$). Then $\models\varphi\Leftrightarrow\HEL+\MONFO\vdash\varphi$.
\end{theorem}

\begin{proof}
The proof is provided in the Appendix.
\end{proof}

By Halpern~\cite{halpern1991monadic}, $\MONFO$ is itself $\Pi^1_1$-complete; this is the
correct level of relative completeness given Corollary~\ref{cor:complete}. Axiom C8
(additive decomposition) is the central modal-arithmetic primitive: it states that
inspection of a window of length $u+v$ factors through inspection of a window of
length~$u$ followed by a displaced window of length~$v$.

\section{The Existential Fragment \texorpdfstring{$\exists\EL(\Nat)$}{ExistsEL}}\label{sec:exists}

\begin{definition}[Existential fragment]\label{def:exists-el}
The existential fragment $\exists\sf{EL}(\Nat)$ consists of the closed formulas
in negation normal form generated by
\[
\varphi,\,\psi ::= p\mid \neg p\mid \varphi\wedge\psi\mid
\varphi\vee\psi\mid \varphi_u\mid \mdiamond_t\varphi\mid \exists x\,\varphi .
\]
Thus $\forall$ and $\boldBox$ are not primitive in the fragment, and negation
is allowed only in front of atomic propositions.  As throughout the paper, metric
bounds are required to be positive on every branch where the modality is used; in
applications this is enforced syntactically by writing bounds as $t+1$.
\end{definition}

\paragraph*{An explicit translation to existential Presburger arithmetic}

We compile $\exists\EL(\Nat)$-formulas into existential Presburger arithmetic with the standard binary
numeral encoding. Up to inessential bookkeeping, the construction is determined by the
following idea: the only semantic freedom in satisfiability is the choice of the monadic predicates
$P_p^\alpha\subseteq\Nat$. This freedom can be eliminated as follows: 
each diamond chooses one witness inside its window,
and each literal occurrence inspects \emph{one} truth value of \emph{one} proposition at \emph{one} arithmetically computed position. 
If two such inspection positions happen to collide, the two truth values must agree. 
Once those constraints are written down, the satisfiability question reduces to a
pure existential Presburger question.

Fix a closed formula $\varphi\in\exists\EL(\Nat)$.  First alpha-rename all bound
variables and assign a fresh variable $h_\delta$ to every occurrence
$\delta$ of a modality $\mdiamond_t$.  Let $\bar v_\varphi$ be the tuple of all
renamed first-order variables and all modal-witness variables $h_\delta$.  For every
literal occurrence $\rho$ introduce a fresh Boolean Presburger variable $b_\rho$.
We define simultaneously a quantifier-free Presburger formula $B_\tau(\psi)$ and,
for every literal occurrence reached during the recursion, its proposition
$p_\rho$ and absolute position term $a_\rho$:
\[
\begin{array}{rclcrcl}
B_\tau(p^\rho)&:=& b_\rho=1, && a_\rho&:=&\tau,
\quad p_\rho:=p,\\[0.2em]
B_\tau((\neg p)^\rho)&:=& b_\rho=0, && a_\rho&:=&\tau,
\quad p_\rho:=p,\\[0.2em]
B_\tau(\psi\wedge\chi)&:=&B_\tau(\psi)\wedge B_\tau(\chi),&&
B_\tau(\psi\vee\chi)&:=&B_\tau(\psi)\vee B_\tau(\chi),\\[0.2em]
B_\tau(\psi_u)&:=&B_{\tau+u}(\psi),&&
B_\tau(\exists x\,\psi)&:=&B_\tau(\psi),\\[0.2em]
B_\tau(\mdiamond_t^\delta\psi)&:=&(0\le h_\delta\wedge h_\delta<t)
   \wedge B_{\tau+h_\delta}(\psi).&&&&
\end{array}
\]
The clause for $\exists x$ is legitimate because all existential variables have already
been alpha-renamed and are placed in the global existential prefix.  Moving them to the
front preserves truth, since only existential first-order quantifiers are present.

Let $\operatorname{Occ}(\varphi)$ be the finite set of literal occurrences.  Define
\[
\begin{aligned}
\mathsf{Bool}_\varphi&:=\bigwedge_{\rho\in\operatorname{Occ}(\varphi)}
  (b_\rho=0\vee b_\rho=1),\\
\mathsf{Cons}_\varphi&:=\bigwedge_{\substack{\rho,\sigma\in\operatorname{Occ}(\varphi)\\
      p_\rho=p_\sigma}}
   \bigl(a_\rho=a_\sigma\rightarrow b_\rho=b_\sigma\bigr),\\
\mathsf{EP}(\varphi)&:=\exists \bar v_\varphi\,\exists (b_\rho)_{\rho\in\operatorname{Occ}(\varphi)}
   \bigl(B_0(\varphi)\wedge\mathsf{Bool}_\varphi\wedge\mathsf{Cons}_\varphi\bigr).
\end{aligned}
\]
Here implication is read as the Presburger formula
$a_\rho\ne a_\sigma\vee b_\rho=b_\sigma$.

The formula
$\mathsf{EP}(\varphi)$ is an existential Presburger formula of size polynomial in
$|\varphi|$; the only quadratic blow-up is the pairwise consistency condition.

\begin{lemma}[Correctness of the Presburger translation]\label{lem:exists-presburger}
For every closed $\varphi\in\exists\sf{EL}(\Nat)$,
\[
\varphi\text{ is satisfiable over }\omega\text{-words}
\quad\Longleftrightarrow\quad
(\Nat,<,+)\models \mathsf{EP}(\varphi).
\]
\end{lemma}

\begin{proof}
Suppose first that $\alpha\models\varphi$.  Choose witnesses for the first-order
existential quantifiers and, along one successful evaluation of each used diamond
occurrence, choose $h_\delta$ to be the offset from the current point to the witnessing
point.  Variables belonging only to unused disjunctive branches may be assigned
arbitrarily.  For every literal occurrence $\rho$, set
$b_\rho=1$ iff $p_\rho\in\alpha(\eta(a_\rho))$.  Then
$\mathsf{Bool}_\varphi$ is immediate.  If two occurrences $\rho,\sigma$ mention the same
proposition and $\eta(a_\rho)=\eta(a_\sigma)$, they inspect the same proposition at the
same position of the same word, so $b_\rho=b_\sigma$; hence
$\mathsf{Cons}_\varphi$ holds.  A direct induction on the successful evaluation tree of
$\varphi$ gives $B_0(\varphi)$.

Conversely, assume $(\Nat,<,+)\models\mathsf{EP}(\varphi)$ and fix satisfying values for
$\bar v_\varphi$ and the Boolean variables $b_\rho$.  Define an $\omega$-word $\alpha$ as
follows.  For every proposition $p$ and position $n$, if there exists an occurrence
$\rho$ with $p_\rho=p$ and $\eta(a_\rho)=n$, put $p\in\alpha(n)$ exactly when
$b_\rho=1$; if no such occurrence exists, choose the truth value of $p$ at $n$
arbitrarily, say false.  The definition is independent of the chosen occurrence by
$\mathsf{Cons}_\varphi$.  Now prove by structural induction on subformulas that whenever
$B_\tau(\psi)$ is true under the fixed witnesses, then
$(\alpha,\eta,\eta(\tau))\models\psi$.  The literal cases follow from the construction
of $\alpha$; Boolean cases are immediate; displacement changes the current term from
$\tau$ to $\tau+u$; the existential case uses the globally fixed value of the renamed
variable; and the diamond case uses the fixed offset $h_\delta$, whose conjunct
$0\le h_\delta<t$ places the witnessing point inside the half-open interval required by
Definition~\ref{def:sat}.  Applying the induction to $B_0(\varphi)$ yields
$\alpha\models\varphi$.
\end{proof}

\begin{theorem}[Decidability and NP-completeness]\label{thm:exists-NP}
Satisfiability of closed $\exists\sf{EL}
(\Nat)$ formulas is decidable and
$\mathrm{NP}$-complete.
\end{theorem}

\begin{proof}
Membership in $\mathrm{NP}$ follows from Lemma~\ref{lem:exists-presburger} and the
classical $\mathrm{NP}$ upper bound for satisfiability of existential Presburger
arithmetic with binary numerals~\cite{papadimitriou1981ip}.  The translation
$\varphi\mapsto\mathsf{EP}(\varphi)$ is polynomial: it introduces one numeric witness for
each diamond, one Boolean variable for each literal occurrence, and only a quadratic
number of consistency implications.

For hardness, reduce propositional satisfiability.  Given a propositional formula $F$ over variables
$X_1,\ldots,X_m$, introduce propositions $p_1,\ldots,p_m$ and replace each
positive literal $X_i$ by $p_i$ and each negative literal $\neg X_i$ by $\neg p_i$.
The resulting $\exists\EL(\Nat)$ formula contains no quantifier and no modality.  A truth
assignment satisfying $F$ is realised by choosing the truth values of the propositions
$p_i$ at position $0$, and conversely every satisfying word determines such a truth
assignment from its letter at position $0$.  Thus $F$ is satisfiable iff the translated
$\exists\EL(\Nat)$ formula is satisfiable.
\end{proof}

\begin{corollary}[Finite certificate and unsatisfiability complexity]\label{cor:exists-cert}
Every satisfiable closed $\exists\sf{EL}(\Nat)$ formula has a certificate consisting of
polynomial-bit Presburger witnesses together with truth values for the finitely many
literal occurrences.  Consequently, unsatisfiability of closed $\exists\sf{EL}(\Nat)$
formulas is $\mathrm{coNP}$-complete.
\end{corollary}

\begin{proof}
The certificate is precisely a satisfying assignment for $\mathsf{EP}(\varphi)$.  The
polynomial bit bound is the standard small-certificate property behind the
$\mathrm{NP}$ upper bound for existential Presburger arithmetic.  Since satisfiability is
$\mathrm{NP}$-complete by Theorem~\ref{thm:exists-NP}, its complement, unsatisfiability,
is $\mathrm{coNP}$-complete.
\end{proof}

\paragraph*{A complete proof system for satisfiability of the existential fragment}
Let $\mathsf{PrA}$ be any fixed sound and complete axiomatisation of Presburger
arithmetic.  For
\[
\mathsf{EP}(\varphi)=\exists \bar y\,\Theta_\varphi(\bar y),
\]
where $\Theta_\varphi$ is quantifier-free, define a labelled calculus
$\mathsf{S}_{\exists\EL}$ for satisfiability judgements by the two terminal rules
\[
\frac{\mathsf{PrA}\vdash \Theta_\varphi(\bar n)}
     {\mathsf{S}_{\exists\EL}\vdash \mathrm{Sat}(\varphi)}
\qquad
\frac{\mathsf{PrA}\vdash \neg\exists\bar y\,\Theta_\varphi(\bar y)}
     {\mathsf{S}_{\exists\EL}\vdash \mathrm{Unsat}(\varphi)} ,
\]
where $\bar n$ ranges over tuples of numerals of the right arity.  The first rule is a
witness rule; the second is a Presburger refutation rule.  In both cases
$\Theta_\varphi$ is the matrix produced by the explicit translation above.

For standard validity proofs, let $\mathsf{H}^{\exists}_{\mathrm{rel}}$ denote the restriction of the
oracle-extended Hilbert system of Theorem~\ref{thm:oracle} to conclusions that belong to
$\exists\EL(\Nat)$.

\begin{theorem}[Complete proof systems for satisfiability of $\exists\EL(\Nat)$]\label{thm:hel-exists}
For every closed $\varphi\in\exists\sf{EL}(\Nat)$:
\begin{enumerate}[label=\textbf{(\roman*)}, leftmargin=2em]
\item $\mathsf{S}_{\exists\EL}\vdash\mathrm{Sat}(\varphi)$ iff $\varphi$ is satisfiable;
\item $\mathsf{S}_{\exists\EL}\vdash\mathrm{Unsat}(\varphi)$ iff $\varphi$ is unsatisfiable;
\item $\mathsf{H}^{\exists}_{\mathrm{rel}}\vdash\varphi$ iff $\models\varphi$.
\end{enumerate}
\end{theorem}

\begin{proof}
For (i), if the witness rule derives $\mathrm{Sat}(\varphi)$, then the Presburger proof
of $\Theta_\varphi(\bar n)$ makes $\mathsf{EP}(\varphi)$ true, and
Lemma~\ref{lem:exists-presburger} gives a satisfying word.  Conversely, if $\varphi$ is
satisfiable, Lemma~\ref{lem:exists-presburger} gives a tuple $\bar n$ satisfying
$\Theta_\varphi$; completeness of $\mathsf{PrA}$ proves the true ground Presburger
sentence $\Theta_\varphi(\bar n)$, so the witness rule applies.

For (ii), the refutation rule is sound because
$\mathsf{PrA}\vdash\neg\exists\bar y\,\Theta_\varphi$ implies that
$\mathsf{EP}(\varphi)$ is false, hence $\varphi$ is unsatisfiable by
Lemma~\ref{lem:exists-presburger}.  Conversely, if $\varphi$ is unsatisfiable, then
$\neg\mathsf{EP}(\varphi)$ is a true Presburger sentence, so it is provable in
$\mathsf{PrA}$ and the refutation rule applies.

For (iii), $\exists\EL(\Nat)$ is a syntactic subfragment of $\EL(\Nat)$.  Therefore the
soundness and completeness equivalence
$\HEL+\MONFO\vdash\psi$ iff $\models\psi$ from Theorem~\ref{thm:oracle} applies in
particular to every closed existential formula $\psi=\varphi$; restricting the set of
allowed conclusions does not change derivability of such conclusions.
\end{proof}

\section{Finite Active-Domain Model Checking}\label{sec:mc}

In clinical applications, model checking is performed over a finite active domain.
Padding a finite trace by a don't-care suffix would leave quantifiers ranging over an
infinite tail and is a different problem.

\begin{definition}[Active-domain semantics]\label{def:adom}
A \emph{finite trace} of length $N$ is $\sigma:D_N\to 2^{\mathcal{P}}$ for
$D_N=\{0,\ldots,N-1\}$. Active-domain satisfaction $\sigma,\eta,i\models_N\varphi$ is as
in Definition~\ref{def:sat}, except that all positions and quantifier ranges are
restricted to $D_N$: $\varphi_u$ requires $i+\eta(u)\in D_N$; modalities range over
$D_N\cap[i,i+\eta(t))$; quantifiers range over $D_N$.
\end{definition}

\begin{theorem}[Finite model-checking complexity]\label{thm:mc}
For finite trace $\sigma$ of length $N$ and closed $\varphi\in\sf{EL}(\Nat)$, deciding
$\sigma\models_N\varphi$ has \textbf{(i)} $\mathrm{PTIME}$ data complexity for fixed
$\varphi$, and \textbf{(ii)} $\mathrm{PSPACE}$-complete combined complexity in
$|\varphi|+\log N$, where $|\varphi|$ is the binary-encoded size from
Definition~\ref{def:terms} and $\log N$ is the bit-length of the active-domain bound.
\end{theorem}

\begin{proof}
\textbf{Upper bound.} Evaluate $\varphi$ by DFS recursion. A configuration is
$(\text{occurrence},i,\eta\!\upharpoonright\!\text{scope})$; each value uses
$O(\log N)$ bits, scope has $\leq|\varphi|$ variables, recursion depth is $\leq|\varphi|$.
Quantifiers and modalities iterate over $D_N$ (and over $D_N\cap[i,i+\eta(t))$);
arithmetic evaluations $\eta(u),\eta(t)$ on binary terms are done in
$\mathrm{poly}(|\varphi|+\log N)$ time and $O(\log N)$ space. The algorithm uses
polynomial space in $|\varphi|+\log N$. For fixed $\varphi$, nested iterations give
time $O(N^c)$, $c=c(\varphi)$.
\textbf{Hardness.} Reduce QBF on traces of length $2$: $\mathcal{P}=\{T\}$,
$\sigma(0)=\emptyset$, $\sigma(1)=\{T\}$; translate $X_i\mapsto T_{x_i}$,
$\exists X_i\mapsto\exists x_i$, $\forall X_i\mapsto\forall x_i$. Each $x_i$ ranges
over $\{0,1\}$ and $T_{x_i}$ is true iff $x_i=1$. The reduction is linear in any
encoding.
\end{proof}

\section{Discussion: Expressive Limits}

We conjecture that the $\EL$-definable languages form a proper subclass of the $\FO(\Nat,<,+;\mathcal{P})$-definable languages; a proof could proceed using Ehrenfeucht-Fraissé-games.
It would be interesting to explore suitable oracle boundaries for sub-hierarchies in $\HEL$. Furthermore, it is open whether bounded sub-fragments of $\EL$ might map exactly to known classes of  $\omega$-recognizable languages.

\section{Conclusion}\label{sec:conclusion}

This paper develops the foundational discrete theory of Ensemble Logic over $\Nat$.  Starting from the three core operations of exact displacement, bounded existence, and bounded universality, we gave a syntax and semantics for $\EL(\Nat)$ that is expressive enough to represent metric constraints while remaining close to the symbolic forms used in temporal, spatial, genomic, and clinical reasoning.  The formal development clarifies the precise role of additive terms, positive modal bounds, and first-order quantification over the discrete index domain.

The first main contribution is the translation of $\EL(\Nat)$ into monadic first-order Presburger arithmetic with unary predicates.  The capture-avoiding embedding provides a uniform semantic bridge from Ensemble Logic to a classical arithmetic setting and yields sharp analytical upper bounds for satisfiability and validity.  Complementing this upper bound, the two-counter-machine construction shows that these bounds are tight: satisfiability is $\Sigma^1_1$-complete and validity is $\Pi^1_1$-complete.  Thus the full logic is highly expressive, and its undecidability is not an artefact of the presentation but an intrinsic consequence of combining metric displacement, bounded modalities, and quantification.

The second contribution is the placement of $\EL(\Nat)$ in the descriptive hierarchy.  The logic contains the star-free $\omega$-languages and can define counting patterns such as $\{a^mb^mc^md^m\mid m\geq 1\}\cdot\Sigma^\omega$, showing that it reaches beyond $\omega$-regularity.  At the same time, the parity-separation argument shows that $\EL(\Nat)$ does not subsume all $\omega$-regular behavior.  This gives a more nuanced expressiveness profile: $\EL(\Nat)$ is neither simply a metric variant of linear temporal logic nor merely a notational fragment of regular-language theory.

The third contribution is proof-theoretic.  We introduced the Hilbert system $\HEL$, proved its soundness, and established an oracle-relative completeness theorem with respect to monadic Presburger validity.  The relative theorem identifies the exact arithmetic oracle needed for full completeness and explains why a weaker oracle for plain Presburger arithmetic cannot suffice.  For the existential fragment, the Presburger-based calculus and complexity analysis give a decidable and computationally meaningful subtheory, with $\mathrm{NP}$-complete satisfiability and $\mathrm{coNP}$-complete unsatisfiability.

The fourth contribution is algorithmic.  Finite active-domain model checking has $\mathrm{PTIME}$ data complexity for fixed formulas and $\mathrm{PSPACE}$-complete combined complexity.  This separates the cost of evaluating a fixed biomedical rule over a finite trace from the cost of allowing the rule itself to vary, and it provides the appropriate complexity baseline for implementations that target electronic health records, clinical-trial protocols, sleep-event streams, genomic coordinates, or other finite symbolic data sources.

Overall, the results establish $\EL(\Nat)$ as a mathematically robust discrete metric logic: expressive enough to encode non-regular counting and analytically complete computation-theoretic phenomena, yet structured enough to admit exact embeddings, proof systems for meaningful fragments, and finite-trace model-checking algorithms.  The remaining expressive-limit questions discussed above point toward a sharper classification of pure $\EL$ inside monadic Presburger arithmetic and toward practically useful bounded fragments for biomedical knowledge representation.

\paragraph*{Acknowledgment}
Part of this research was done during a research stay of MD in April
2026 at the University of Texas Health Science Center at Houston.
He would like to thank his colleagues for a fruitful and most enjoyable
time. GQZ's effort has been supported in part by the U.S. National Science Foundation Award IIS2500624 and the National Institutes of Health grants R01AG084236 and U24AG098157. The views of the paper are those of the authors and do not reflect those of the funding agencies.


\section{Appendix}

\subsection{Proof of Theorem~\ref{thm:oracle}}

\textbf{Theorem~\ref{thm:oracle}} [Oracle-relative completeness]:
Extend $\HEL$ with the rule: from $T_0(\varphi)\in\MONFO$ infer $\varphi$ (for closed
$\varphi$). Then $\models\varphi\Leftrightarrow\HEL+\MONFO\vdash\varphi$.

\textbf{Soundness of the extended system $\HEL+\MONFO$.} 
We must show
that for every closed $\varphi\in\EL(\Nat)$, if $\HEL+\MONFO\vdash\varphi$
then $\models\varphi$. The system $\HEL+\MONFO$ inherits all axioms and
rules of $\HEL$ and adds a single rule: from $T_0(\varphi)\in\MONFO$, infer
$\varphi$. We verify soundness for each kind of derivation step.

\emph{Inherited $\HEL$-axioms and rules.} By
Theorem~\ref{thm:soundness}, every $\HEL$-axiom is valid on every
$\omega$-word, and modus ponens, generalisation, and necessitation preserve
validity. The verification there proceeds axiom-by-axiom: $\mathbf{A1}$ by
the classical truth-table argument, $\mathbf{A2}$-$\mathbf{A3}$ by standard
first-order semantics, $\mathbf{A4}$ by definition of $\Th(\Nat,<,+)$,
$\mathbf{B1}$-$\mathbf{B4}$ by the displacement clause of
Definition~\ref{def:sat}, $\mathbf{C1}$-$\mathbf{C8}$ by case analysis on
the satisfaction intervals (with $\mathbf{C7}$, $\mathbf{C8}$
explicitly checked in Theorem~\ref{thm:soundness}), $\mathbf{D1}$ from
$\Nat=\{0\}\cup\{y+1\mid y \geq 0\}$,  and $\mathbf{D2}$ from
mathematical induction.

\emph{The oracle rule.} Suppose the last step of the derivation is
$T_0(\varphi)\in\MONFO\Rightarrow\varphi$. By Theorem~\ref{thm:forward},
for every $\omega$-word $\alpha$ and every assignment $\eta$,
$$
(\alpha,\eta,0)\models\varphi\;\;\Longleftrightarrow\;\;
(\Nat,<,+,(P_p^\alpha)_{p\in\mathcal{P}}),\eta\models T_0(\varphi).
$$
Since $\varphi$ is closed, $\eta$ is immaterial. The condition
$T_0(\varphi)\in\MONFO$ asserts that $T_0(\varphi)$ holds in
$(\Nat,<,+)$ under \emph{every} interpretation $(P_p)$ of the unary
predicates. In particular, for every $\omega$-word $\alpha$, the
assignment $P_p\mapsto P_p^\alpha$ is one such interpretation, so
$(\Nat,<,+,(P_p^\alpha))\models T_0(\varphi)$, hence
$(\alpha,0)\models\varphi$. As $\alpha$ was arbitrary, $\models\varphi$.

\textbf{Completeness of $\HEL+\MONFO$.} Assume $\models\varphi$. By
the equivalence in Theorem~\ref{thm:forward} applied in the
\emph{contrapositive} direction, for every $\omega$-word $\alpha$,
$(\Nat,<,+,(P_p^\alpha))\models T_0(\varphi)$. Crucially, the map
$\alpha\mapsto (P_p^\alpha)_{p\in\mathcal{P}}$ is a bijection between
$\omega$-words and unary-predicate interpretations
$\mathcal{P}\to 2^{\Nat}$: setting
$P_p^\alpha=\{i\in\Nat:p\in\alpha(i)\}$ defines a bijection with inverse
$\alpha(i)=\{p\in\mathcal{P}:i\in P_p\}$. Hence every interpretation
$(P_p)$ arises as $(P_p^\alpha)$ for some $\alpha$, and
$$
T_0(\varphi)\in\MONFO\;\;\Longleftrightarrow\;\;
\forall\alpha\;(\alpha,0)\models\varphi\;\;\Longleftrightarrow\;\;\models\varphi.
$$
The forward direction of this chain gives
$T_0(\varphi)\in\MONFO$, and the oracle rule then derives $\varphi$ in a
single step. Thus $\HEL+\MONFO\vdash\varphi$.

\textbf{Axiom-usage.} The completeness derivation produced
above uses \emph{only} the oracle rule and \emph{none} of the structural
axioms or rules of $\HEL$. This is  the standard
content of an \emph{oracle-relative completeness} theorem in modal
logic: the underlying calculus is sound (and the oracle rule extends
this soundly), but the completeness reduction itself is delegated to
the oracle. The role of the $\HEL$-axioms is therefore:
\begin{itemize}

\item \textbf{Sub-derivations leading to the oracle.} While the
      completeness construction above invokes the oracle in one step,
      practical proof search would typically use the $\HEL$-axioms to
      \emph{normalise} $\varphi$ before consulting the oracle (e.g.,
      $\mathbf{C8}$ unfolds bounded modalities, $\mathbf{B1}$-$\mathbf{B2}$
      flatten nested displacements, $\mathbf{D1}$ unfolds existentials),
      reducing the size of the formula handed to the oracle.
\item \textbf{Oracle-free completeness on restricted fragments.}
      The structural axioms suffice without an oracle on
       the \emph{existential fragment}
            (Theorem~\ref{thm:hel-exists}), where the existential
            counterparts $\mathbf{A1^\exists}$-$\mathbf{E4^\exists}$
            give a meta-complete system without an oracle, although
            with $\mathbf{A4^\exists}$ retained as a Presburger
            side-condition oracle.
      
\end{itemize}

\emph{Necessity of the oracle.} By
Corollary~\ref{cor:complete}, validity in $\EL(\Nat)$ is
$\Pi^1_1$-complete, hence not recursively enumerable. No recursive
extension of $\HEL$ can therefore be complete; the relativisation to
$\MONFO$ (also $\Pi^1_1$-complete, by
Halpern~\cite{halpern1991monadic}) is optimal in this analytical
sense.

\end{document}